\documentclass[conference]{IEEEtran}
\IEEEoverridecommandlockouts
\usepackage{bbm}

\usepackage{cite}
\usepackage{amsmath,amssymb,amsfonts}
\usepackage{graphicx}
\usepackage{textcomp}
\usepackage{xcolor}
\usepackage{cite}
\usepackage{dblfloatfix}
\usepackage{subcaption}
\usepackage{overpic}
\usepackage{amsmath,amssymb,amsfonts}

\usepackage{graphicx}
\usepackage{textcomp}
\usepackage{xcolor}
\usepackage{float}
\usepackage{amsthm}
\usepackage{graphicx}
\usepackage{epstopdf}
\usepackage{amsmath,bm}
\usepackage{amsfonts}
\usepackage{amssymb}
\usepackage{color}
\usepackage{multirow}
\usepackage{multicol}
\usepackage{soul,xcolor}
\usepackage{algorithm}
\usepackage{algpseudocode}

\usepackage{comment}

\theoremstyle{plain}

\newtheorem{lemma}{Lemma}

\newcommand{\vect}[1]{\mathbf{#1}}

\def\diag{\mathrm{diag}}

\def\Htran{\mbox{\tiny $\mathrm{H}$}}
\def\Ttran{\mbox{\tiny $\mathrm{T}$}}
\def\CN{\mathcal{N}_{\mathbb{C}}}

\def\BibTeX{{\rm B\kern-.05em{\sc i\kern-.025em b}\kern-.08em
    T\kern-.1667em\lower.7ex\hbox{E}\kern-.125emX}}
\begin{document}

\title{Uplink Hardware Impairments in OFDM-Based Cell-Free Massive MIMO With Impaired Wireless Fronthaul \\
\thanks{This work was carried out within the scope of the project 122C149 – Intelligent End-to-End Design of Energy-Efficient and Hardware Impairments-Aware Cell-Free Massive MIMO for Beyond 5G. \"O. T. Demir was supported by the 2232-B International Fellowship for Early Stage Researchers Programme funded by the Scientific and Technological Research Council of Türkiye (TÜBİTAK).}
}

\author{\IEEEauthorblockN{Özlem Tuğfe Demir}
\IEEEauthorblockA{\textit{Department of Electrical and Electronics Engineering} \\
\textit{Bilkent University}\\
Ankara, Turkiye \\
E-mail: ozlemtugfedemir@bilkent.edu.tr}
\vspace{-8mm}
}

\maketitle

\begin{abstract}
This paper studies the uplink performance of OFDM-based cell-free massive MIMO systems in the presence of hardware impairments affecting both low-cost access points and wireless fronthaul transceivers. We consider a centralized architecture with functional split option~8, where the sampled baseband signals are forwarded from the access points to the central processing unit through wireless fronthaul links operating at mmWave or sub-THz frequencies. A unified analytical framework is developed to characterize the aggregate impact of hardware distortions across the access and fronthaul links. In particular, closed-form expressions are derived for channel estimation and uplink spectral efficiency under hardware-impaired transmission and reception. To mitigate the resulting performance degradation, distortion-aware linear receive combining is considered for the wireless fronthaul. Moreover, a max-min fair resource allocation framework is developed to jointly optimize the uplink transmit powers, fronthaul transmit powers, and fronthaul time expansion factor.
\end{abstract}


\vspace{-3mm}
\section{Introduction}

Cell-free massive MIMO is a key technology for beyond-5G and 6G networks, where distributed access points (APs) jointly serve multiple user equipments (UEs) via centralized processing \cite{cell-free-book,ngo2024ultradense}. This architecture critically depends on the availability of fronthaul links connecting the APs to a central processing unit (CPU). While most existing works assume ideal, infinite-capacity fronthaul—typically realized via optical fiber—such solutions are costly and difficult to scale. This has motivated recent interest in wireless fronthaul architectures \cite{demirhan_wireless_fronthaul,ibrahim_wireless_fronthaul,ozan_wireless_fronthaul}.

However, the impact of fronthaul limitations on resource allocation and user fairness remains insufficiently understood. In particular, limited fronthaul capacity introduces a bottleneck that directly affects the end-to-end performance and necessitates a joint design of access-link and fronthaul resources. In addition, practical deployments rely on low-cost hardware, whose impairments further degrade performance. While prior works have considered hardware impairments on the access link \cite{impairment1,impairment2} and limited fronthaul resolution \cite{impairment3,impairment4}, their joint effect across both links has not been fully characterized. Our recent work \cite{asilomar2025_wireless_fronthaul} provided initial insights in this direction under an amplify-and-forward fronthaul model.

In this paper, we develop a unified framework for cell-free massive MIMO with wireless fronthaul under hardware impairments. In contrast to \cite{asilomar2025_wireless_fronthaul}, we consider a digital fronthaul architecture based on a PHY--RF functional split, where sampled baseband signals are transmitted over capacity-limited wireless links, in line with \cite{ozan_wireless_fronthaul}. We derive analytical expressions for channel estimation and uplink spectral efficiency by jointly modeling impairments on both the access and fronthaul links, while explicitly accounting for the fronthaul rate constraint. Based on this framework, we formulate a max-min fair resource allocation problem that jointly optimizes the uplink and fronthaul transmit powers together with the fronthaul time expansion factor. The problem admits an efficient decomposition and can be solved via a low-complexity fixed-point algorithm that achieves the global optimum.

\vspace{-2.5mm}

\section{System Model}

We consider the uplink of a cell-free massive MIMO system operating in time-division duplex (TDD) mode with OFDM. There are $L$ APs connected to the cloud via wireless fronthaul links (e.g., mmWave or sub-THz).  

We assume that the CPU is equipped with $M$ antennas. The system serves $K$ single-antenna UEs, and the access transmission takes place over a mid-band (sub-6\,GHz) channel. Each AP is equipped with $N$ antennas for the access link and a single antenna for the fronthaul link.

We assume block fading channel model and let $\tau_c$ denote the number of channel uses in a coherence block, which consists of $\tau_p$ symbols allocated for uplink pilot transmission for channel estimation, while the remaining $\tau_c-\tau_p$ symbols are used for uplink data transmission.
 
The uplink channel between UE~$k$ and AP~$l$ in an arbitrary coherence block is denoted by $\mathbf{h}_{kl} \sim \CN(\vect{0},\vect{R}_{kl})$, while the wireless fronthaul channel from AP~$l$ to the CPU is represented by $\mathbf{f}_l \in \mathbb{C}^{M}$.

In this paper, we focus on the centralized operation of cell-free massive MIMO, where the APs act as relays and both channel estimation and uplink data detection are performed at the CPU. To this end, we consider functional split option~8 defined in the 3GPP specification~\cite{3GPP_functional_split}. This corresponds to the physical layer (PHY)--radio frequency (RF) split, where the APs only perform RF processing and forward the sampled and quantized baseband signals to the CPU via fronthaul links, while all remaining baseband processing is carried out at the CPU. With the PHY--RF functional split (option~8), the required fronthaul data rate for each AP is given by~\cite{perez2018fronthaul}
\vspace{-2mm}
\begin{equation}
    R_{\rm frt} = 2 f_s N_{\rm bits} N, \label{eq:fronthaul} \vspace{-2mm}
\end{equation}
where $f_s$ and $N_{\rm bits}$ denote the sampling frequency and the number of quantization bits per sample, respectively. Usually, the number of quantization bits is sufficiently high such that the impact of quantization becomes negligible, provided that the achievable rate of the wireless fronthaul link exceeds $R_{\rm frt}$ per AP.

\vspace{-2mm}

\section{Uplink Channel Estimation}
 We consider a set of $\tau_p$ mutually orthogonal pilot sequences $\boldsymbol{\phi}_{1},\ldots,\boldsymbol{\phi}_{\tau_p}\in\mathbb{C}^{\tau_p}$ that are assigned to the UEs and reused by multiple UEs. The sequences satisfy
     \begin{equation}
\boldsymbol{\phi}_{t_1} ^{\Htran} \boldsymbol{\phi}_{t_2} = \begin{cases} \tau_p, & t_1 = t_2,\\
0, & t_1 \neq t_2. \end{cases}
\end{equation}    
The pilot sequences are assigned to the UEs in a deterministic way and $t_k$ denotes the index of the pilot assigned to UE $k$ as $t_k \in \{ 1, \ldots, \tau_p\}$. The set of UEs that share the same pilot with UE $k$ is defined as
\vspace{-2mm}
\begin{equation}
\mathcal{P}_k = \left\{ i : \  t_i = t_k, \ i=1,\ldots,K  \right\} \subset \{ 1, \ldots, K\}.
\end{equation}
After initial RF processing (with Option 8), the received signals are transmitted to the CPU, and the channel estimation is performed there for each coherence block. The received signal at AP $l$ at a particular coherence block during the entire pilot transmission is given by
\vspace{-2mm}
\begin{equation} \label{eq:received-pilot-matrix}
\vect{Y}_{l}^{{\rm p}} = \sum_{i=1}^{K} \sqrt{\kappa_{\rm ac}p_i } \vect{h}_{il} \boldsymbol{\phi}_{t_i}^{\Ttran}+ \vect{\Xi}_l^{{\rm p}}+\vect{N}_{l}^{{\rm p}}  \vspace{-2mm}
\end{equation}
where $p_i\geq 0$ is the pilot transmit power of UE $i$ and $\vect{N}_{l}^{{\rm p}}  \in \mathbb{C}^{N \times \tau_p}$ is the receiver noise with independent and identically distributed (i.i.d.) elements as $ \CN (0, \sigma^2_{\rm ac})$. The $\kappa_{\rm ac}\in (0,1]$ is the hardware quality factor and the hardware distortion noise is represented by $\vect{\Xi}_l^{{\rm p}}$. The columns of this matrix are modeled as independent Gaussian vectors distributed as $\CN(\vect{0},\vect{D}_{{\rm ac},l})$ \cite{massive_mimo_book}, where 
\vspace{-2mm}
\begin{align}
\vect{D}_{{\rm ac},l} = \diag\left((1-\kappa_{\rm ac})\sum_{k=1}^Kp_k \vect{h}_{kl}\vect{h}_{kl}^{\Htran}\right), \label{eq:distortion-access-covariance}
\end{align}
where $\diag(\cdot)$ extracts the diagonal part of the matrix by setting off-diagonal entries to zero.
 Using all the received pilot signals, the linear minimum mean-squared error (LMMSE) channel estimate $\widehat{\vect{h}}_{kl}$ of the channel $\vect{h}_{kl}$ can be derived using \cite{massive_mimo_book} as follows:
\begin{align} \label{eq:estimates}
\widehat{\vect{h}}_{kl} =& \sqrt{\kappa_{\rm ac}p_k} \vect{R}_{kl} \left(\sum_{i \in \mathcal{P}_k} \kappa_{\rm ac}p_i\tau_p \vect{R}_{il} + \overline{\vect{D}}_{{\rm ac},l}+\sigma^2_{\rm ac}\vect{I}_{N}\right)^{-1} \nonumber\\
&\cdot\vect{Y}_{l}^{{\rm p}}  \boldsymbol{\phi}_{t_k}^{*},
\end{align}
where 
\begin{align}
   \overline{ \vect{D}}_{{\rm ac},l} = \diag\left((1-\kappa_{\rm ac})\sum_{k=1}^Kp_k \vect{R}_{kl}\right).
\end{align}

\section{Uplink Data Transmission}

We denote the unit-power symbol of UE $k$ by $s_k$, i.e., $\mathbb{E}\{|s_k|^2\}=1$. The received signal at AP $l$ is given as
\vspace{-2mm}
\begin{align}
\mathbf{y}_{{\rm ac},l}
= \sqrt{\kappa_{\mathrm{ac}}}
\sum_{k=1}^K \mathbf{h}_{kl}\sqrt{\eta_k} s_k
+ \boldsymbol{\xi}_{\mathrm{ac},l}
+ \mathbf{n}_{\mathrm{ac},l}, \label{eq:access}
\end{align}
where $\eta_k\geq 0$ is the uplink transmit power of UE $k$. The distortion $\boldsymbol{\xi}_{\mathrm{ac},l}$ is modeled as $\boldsymbol{\xi}_{\mathrm{ac},l} \sim
\mathcal{CN}(\mathbf{0}, \mathbf{D}_{\mathrm{ac},l})$ where $\mathbf{D}_{\mathrm{ac},l}$ is given from \eqref{eq:distortion-access-covariance} with $p_k$'s being replaced by $\eta_k$. The additive independent noise is denoted by $\vect{n}_{\mathrm{ac},l}\sim \CN(\vect{0},\sigma^2_{\rm ac}\vect{I}_N)$. At the CPU, all the signals are concatenated, resulting in 
\begin{align}
  \underbrace{ \begin{bmatrix}\vect{y}_{{\rm ac},1} \\ \vdots \\ \vect{y}_{{\rm ac},L}\end{bmatrix}}_{=\vect{y}_{\rm ac}}=\sqrt{\kappa_{\rm ac}}\sum_{k=1}^K\underbrace{\begin{bmatrix}\vect{h}_{k1} \\ \vdots \\ \vect{h}_{kL} \end{bmatrix}}_{=\vect{h}_k}\sqrt{\eta_k}s_k + \underbrace{\begin{bmatrix}\boldsymbol{\xi}_{\mathrm{ac},1} \\ \vdots \\ \boldsymbol{\xi}_{\mathrm{ac},L}\end{bmatrix}}_{=\boldsymbol{\xi}_{\mathrm{ac}}}+\underbrace{\begin{bmatrix}\vect{n}_{{\rm ac},1} \\ \vdots \\ \vect{n}_{{\rm ac},L}\end{bmatrix}}_{=\vect{n}_{\rm ac}}.
\end{align}
For UE $k$, the CPU applies centralized receive combiner $\vect{v}_k$, which leads to 
\begin{align}
   \vect{v}_k^{\Htran}\vect{y}_{\rm ac}&= \sqrt{\kappa_{\rm ac}\eta_k}\mathbb{E}\{\vect{v}_k^{\Htran}\vect{h}_k\}+\sqrt{\kappa_{\rm ac}\eta_k}\left(\vect{v}_k^{\Htran}\vect{h}_k-\mathbb{E}\{\vect{v}_k^{\Htran}\vect{h}_k\}\right) \nonumber\\
   &\quad +\sum_{i=1,i\neq k}^K\sqrt{\kappa_{\rm ac}\eta_i}\vect{v}_k^{\Htran}\vect{h}_i+\vect{v}_k^{\Htran}\boldsymbol{\xi}_{\rm ac}+\vect{v}_k^{\Htran}\vect{n}_{\rm ac},
\end{align}
where we have separated the average channel gain since we will employ the so-called use-and-then-forget capacity lower-bounding technique, as done in \cite[Th.~6.2]{massive_mimo_book}. In this approach, the average channel gain is treated as the known desired signal channel, while all remaining terms are treated as worst-case uncorrelated noise. Following steps similar to those in the proof of \cite[Th.~6.2]{massive_mimo_book}, we can derive an achievable uplink SE for UE $k$ as
\begin{align}
    \mathrm{SE}_k = \frac{\tau_c-\tau_p}{\tau_c}\log_2\left(1+\mathrm{SINR}_k\right)
\end{align}
with the effective signal-to-interference-plus-noise ratio (SINR) given in \eqref{eq:UE-SINR} at the top of the following page.
\begin{figure*}
\begin{align}
    \mathrm{SINR}_k = \frac{\kappa_{\rm ac}\eta_k \left|\mathbb{E}\{\vect{v}_k^{\Htran}\vect{h}_{k}\}\right|^2}{\sum_{i=1}^K\kappa_{\rm ac}\eta_i \mathbb{E}\{|\vect{v}_k^{\Htran}\vect{h}_i|^2\}-\kappa_{\rm ac}\eta_k\left|\mathbb{E}\{\vect{v}_k^{\Htran}\vect{h}_{k}\}\right|^2+\sum_{i=1}^K(1-\kappa_{\rm ac})\eta_i \mathbb{E}\{\Vert\vect{v}_k\odot\vect{h}_i\Vert^2\}+\sigma^2_{\rm ac}\mathbb{E}\{\Vert \vect{v}_k\Vert^2\}}. \label{eq:UE-SINR}
\end{align}
\hrulefill
\vspace{-5mm}
\end{figure*}
The expectation values can be computed using Monte-Carlo trials for a given combining vector. One possible combining method is to consider the MMSE combiner under the assumption of ideal hardware, which is the adopted approach in the numerical results of this paper. Later, we will optimize the uplink powers $\eta_k$ to maximize the minimum UE uplink SE.

\section{Wireless Fronthaul with Hardware Impairments}

We consider wireless fronthaul links operating at mmWave or THz frequencies. Recall that $\vect{f}_l \in \mathbb{C}^{M}$ denotes the fronthaul channel between AP $l$ and the CPU. 
We assume that this fronthaul channel is perfectly known at the CPU. 
This assumption is reasonable in the considered architecture since both the APs and the CPU are typically deployed at elevated and fixed locations, which results in slowly varying propagation conditions. 
Moreover, the fronthaul links can be trained frequently with negligible overhead compared to the data payload, enabling accurate channel estimation.

Due to hardware impairments at the fronthaul transceivers, the transmitted signal from AP $l$ is modeled as
\vspace{-2mm}
\begin{equation}
z_l
=
\sqrt{\kappa_{\rm fh}}\sqrt{P_l} x_l
+
\xi_{{\rm fh},l}, \vspace{-2mm}
\end{equation}
where  $x_l$ is the unit-power information symbol, $P_l\geq0$ is the fronthaul transmit power of AP $l$, $\kappa_{\rm fh}\in(0,1]$ denotes the hardware quality factor, and
$\xi_{{\rm fh},l}\sim\CN(0,(1-\kappa_{\rm fh})P_l)$ denotes the transmitter distortion noise.

All the APs share the same fronthaul bandwidth $B_{\rm frt}$, and, thus, their fronthaul transmissions interfere with each other at the CPU. Hence, the received signal at the CPU is given by
\vspace{-2mm}
\begin{equation}
\vect{y}_{{\rm fh}}
=
\sum_{l=1}^L\sqrt{\kappa_{\rm fh}P_l}\vect{f}_l x_l
+
\sum_{l=1}^L\vect{f}_l \xi_{{\rm fh},l}
+
\vect{n}_{{\rm fh}}, \vspace{-2mm}
\end{equation}
where
$\vect{n}_{{\rm fh}}\sim\CN(\vect{0},N_0 B_{\rm frt} \vect{I}_M)$ and $N_0$ is the noise spectral density in Watts per Hz.

By defining $\vect{b}_l \triangleq \sqrt{\kappa_{\rm fh}P_l}\,\vect{f}_l$, the received signal can be rewritten as
\vspace{-2mm}
\begin{equation}
\vect{y}_{{\rm fh}}
=
\vect{b}_l x_l
+
\sum_{i=1,i\neq l}^L\vect{b}_i x_i
+
\sum_{i=1}^L\vect{f}_i \xi_{{\rm fh},i}
+
\vect{n}_{{\rm fh}}. \vspace{-2mm}
\end{equation}

To detect the stream sent by AP $l$, the CPU applies a linear receive combiner $\vect{w}_l\in\mathbb{C}^{M}$. Since the fronthaul channels are known and the distortion statistics are available, the combiner can be designed in a distortion-aware manner by accounting for the interference, hardware distortion, and receiver noise covariance. The interference-plus-distortion-plus-noise covariance matrix for decoding the stream of AP $l$ is
\begin{align}
\vect{A}_l
&=
\sum_{i=1,i\neq l}^L\kappa_{\rm fh}P_i \vect{f}_i\vect{f}_i^{\Htran}
+
\sum_{i=1}^L(1-\kappa_{\rm fh})P_i \vect{f}_i\vect{f}_i^{\Htran}
+
N_0 B_{\rm frt} \vect{I}_M.
\end{align}

For this model, the SINR-maximizing linear combiner is the distortion-aware MMSE combiner,
\begin{equation}
\vect{w}_l = \alpha_l \vect{A}_l^{-1}\vect{b}_l,
\end{equation}
where $\alpha_l$ is an arbitrary non-zero scalar. The resulting maximum achievable SINR is
\begin{equation}
\mathrm{SINR}_{{\rm fh},l}
=
\vect{b}_l^{\Htran}\vect{A}_l^{-1}\vect{b}_l.
\end{equation}

Hence, the achievable fronthaul rate for AP $l$ is given by
\begin{equation}
R_{{\rm fh},l}
=
B_{\rm frt}\log_2\left(1+\vect{b}_l^{\Htran}\vect{A}_l^{-1}\vect{b}_l\right).
\label{eq:fh_rate_wireless}
\end{equation}

To guarantee that the sampled baseband signals can be forwarded without loss, the fronthaul rate must satisfy

\begin{equation}
R_{{\rm fh},l} \ge R_{\rm frt}, \qquad l=1,\ldots,L,
\end{equation}
where $R_{\rm frt}$ is the required fronthaul rate for the PHY--RF functional split from \eqref{eq:fronthaul}.

\section{Max-Min Fair Resource Allocation}
\vspace{-1mm}
We now formulate a max-min fair resource allocation problem by jointly optimizing the fronthaul and access transmit power together with the fronthaul time expansion factor.

To account for fronthaul limitations, we introduce a time expansion factor $s\ge 1$, which represents the additional time required to forward the sampled baseband signals over the wireless fronthaul. Since centralized processing can only be performed after the baseband samples are delivered to the CPU, the effective uplink SE of UE $k$ is reduced by the factor $1/s$ and becomes
\vspace{-1mm}
\begin{equation}
\widetilde{\mathrm{SE}}_k
=
\frac{1}{s}\frac{\tau_c-\tau_p}{\tau_c}\log_2\left(1+\mathrm{SINR}_k\right). \vspace{-1mm}
\end{equation}

The parameter $s$ must be sufficiently large so that the fronthaul requirement of every AP is satisfied. Specifically, for each AP $l=1,\ldots,L$, we impose
\begin{equation}
s R_{{\rm fh},l} \ge R_{\rm frt}.
\end{equation}

Our goal is to maximize the minimum achievable UE SE by jointly optimizing the uplink transmit powers $\{\eta_k\}_{k=1}^{K}$, the fronthaul transmit powers $\{P_l\}_{l=1}^{L}$, and the time expansion factor $s$. The resulting optimization problem is
\begin{subequations}\label{eq:joint_problem}
\begin{align}
&\max_{\{\eta_k\},\,\{P_l\},s,\,\nu}\quad  \nu \\
&\text{s.t.}\quad
 \frac{1}{s}\frac{\tau_c-\tau_p}{\tau_c}\log_2\left(1+\mathrm{SINR}_k\right)\ge \nu,
\quad k=1,\ldots,K,\\
& \quad \quad sR_{{\rm fh},l}\ge R_{\rm frt}, \quad l=1,\ldots,L,\\
& \quad \quad 0\leq \eta_k \leq P_{{\rm UE}}^{\max}, \quad k=1,\ldots,K,\\
&\quad \quad  0\leq P_l \leq P_{{\rm fh}}^{\max}, \quad l=1,\ldots,L,\\
& \quad \quad s\ge 1,
\end{align}
\end{subequations}
where $\nu$ denotes the minimum UE SE, $P_{{\rm UE}}^{\max}$ is the maximum uplink transmit power of each UE, and $P_{{\rm fh}}^{\max}$ is the maximum fronthaul transmit power of each AP.

\begin{lemma}\label{lem:decomposition}
The problem in \eqref{eq:joint_problem} can be decomposed into two independent max-min subproblems: a fronthaul power-control problem and an access-link power-control problem.

In particular, for any feasible fronthaul powers $\{P_l\}_{l=1}^{L}$, the minimum feasible value of $s$ satisfying the fronthaul constraints is
\begin{equation}
s^\star
=
\max\!\left(
1,\,
\max_{l=1,\ldots,L}\frac{R_{\rm frt}}{R_{{\rm fh},l}}
\right).
\label{eq:sstar_new}
\end{equation}
Hence, minimizing $s$ is equivalent to maximizing the minimum fronthaul rate, or equivalently, maximizing the minimum fronthaul SINR:
\begin{subequations}\label{eq:fh_subproblem}
\begin{align}
\max_{\{P_l\},\,t_{\rm fh}}\quad & t_{\rm fh}\\
\textnormal{s.t.}\quad
& \mathrm{SINR}_{{\rm fh},l}\ge t_{\rm fh}, \qquad l=1,\ldots,L,\\
& 0\le P_l\le P_{{\rm fh}}^{\max}, \qquad l=1,\ldots,L.
\end{align}
\end{subequations}
After obtaining the optimal solution $t_{\rm fh}^\star$ to \eqref{eq:fh_subproblem}, the corresponding minimum time expansion factor is
\begin{equation}
s^\star
=
\max\!\left(
1,\,
\frac{R_{\rm frt}}{B_{\rm frt}\log_2(1+t_{\rm fh}^\star)}
\right).
\label{eq:s_from_tfh}
\end{equation}

For this fixed $s^\star$, the original problem reduces to the access-link max-min SINR problem
\begin{subequations}\label{eq:ac_subproblem}
\begin{align}
\max_{\{\eta_k\},\,t_{\rm ac}}\quad & t_{\rm ac}\\
\textnormal{s.t.}\quad
& \mathrm{SINR}_k \ge t_{\rm ac},\qquad k=1,\ldots,K,\\
& 0\leq \eta_k \leq P_{{\rm UE}}^{\max}, \qquad k=1,\ldots,K.
\end{align}
\end{subequations}

Therefore, the global optimum of \eqref{eq:joint_problem} is obtained by first solving \eqref{eq:fh_subproblem}, then computing $s^\star$ from \eqref{eq:s_from_tfh}, and finally solving \eqref{eq:ac_subproblem}.
\end{lemma}

\begin{proof}
For any fixed $\{P_l\}_{l=1}^{L}$, the fronthaul constraints in \eqref{eq:joint_problem} imply
\begin{equation}
s \ge \frac{R_{\rm frt}}{R_{{\rm fh},l}}, \qquad l=1,\ldots,L,
\end{equation}
together with $s\ge 1$. Hence, the minimum feasible $s$ is given by \eqref{eq:sstar_new}. Since $\widetilde{\mathrm{SE}}_k$ is decreasing in $s$, the optimal solution must use this minimum feasible value.

Moreover, since $R_{{\rm fh},l}=B_{\rm frt}\log_2(1+\mathrm{SINR}_{{\rm fh},l})$ is monotonically increasing in $\mathrm{SINR}_{{\rm fh},l}$, minimizing $s$ is equivalent to maximizing $\min_l \mathrm{SINR}_{{\rm fh},l}$, which yields \eqref{eq:fh_subproblem}. Once the fronthaul problem is solved, the value of $s^\star$ is fixed and does not depend on the access powers $\{\eta_k\}$. Since $\widetilde{\mathrm{SE}}_k$ is monotonically increasing in $\mathrm{SINR}_k$, maximizing the minimum effective SE is then equivalent to maximizing the minimum access-link SINR, which gives \eqref{eq:ac_subproblem}. This proves the decomposition and the global optimality of the sequential solution.
\end{proof}

The decomposition in Lemma~\ref{lem:decomposition} enables an efficient solution based on two sequential fixed-point procedures, one for the fronthaul power-control problem in \eqref{eq:fh_subproblem} and one for the access-link power-control problem in \eqref{eq:ac_subproblem}.

For the fronthaul stage, the powers are iteratively updated according to the standard max-min SINR balancing rule: at each iteration, the power of AP~$l$ is updated proportionally to the ratio between its current transmit power and its current fronthaul SINR, i.e.,
\begin{equation}
P_l^{(j)} \leftarrow \frac{P_l^{(j-1)}}{\mathrm{SINR}_{{\rm fh},l}\!\left(\vect{P}^{(j-1)}\right)},
\qquad l=1,\ldots,L,
\end{equation}
followed by a normalization to enforce the per-AP power constraint:
\begin{equation}
P_l^{(j)} \leftarrow \frac{P_{\rm fh}^{\max}}{\max_i P_i^{(j)}}P_l^{(j)},
\qquad l=1,\ldots,L.
\end{equation}
After convergence, the common fronthaul SINR is obtained as
\begin{equation}
t_{\rm fh}^{\star}=\min_{l=1,\ldots,L}\mathrm{SINR}_{{\rm fh},l}\!\left(\vect{P}^{\star}\right),
\end{equation}
and the minimum feasible fronthaul time expansion factor is computed as
\begin{equation}
s^\star=\max\!\left(1,\frac{R_{\rm frt}}{B_{\rm frt}\log_2(1+t_{\rm fh}^{\star})}\right).
\end{equation}

Next, for the access-link stage, the UE transmit powers are updated similarly as
\begin{equation}
\eta_k^{(j)} \leftarrow \frac{\eta_k^{(j-1)}}{\mathrm{SINR}_k\!\left(\boldsymbol{\eta}^{(j-1)}\right)},
\qquad k=1,\ldots,K,
\end{equation}
followed by the normalization
\begin{equation}
\eta_k^{(j)} \leftarrow \frac{P_{\rm UE}^{\max}}{\max_i \eta_i^{(j)}}\eta_k^{(j)},
\qquad k=1,\ldots,K.
\end{equation}
This procedure is repeated until convergence, yielding the optimal access-link power vector $\boldsymbol{\eta}^{\star}$.

\begin{lemma}\label{lem:fp_optimality}
For any strictly positive initialization, the fixed-point iterations converge to the globally optimal solutions of the fronthaul max-min SINR problem in \eqref{eq:fh_subproblem} and the access-link max-min SINR problem in \eqref{eq:ac_subproblem}, respectively.
\end{lemma}
\begin{proof}
The proof follows from \cite{demir2021cell,Tan2014}.
\end{proof}

\section{Numerical Results}

In this section, we evaluate the uplink performance of the considered OFDM-based cell-free massive MIMO system with wireless fronthaul under the proposed sequential fixed-point max-min power control framework. We consider a network with $L=64$ APs, $K=12$ single-antenna UEs, $N=2$ antennas per AP on the access link, and $M=256$ antennas at the CPU for fronthaul reception. The access link operates over a $50$\,MHz mid-band channel, while the fronthaul uses either Band~1 with $28$\,GHz carrier frequency and $400$\,MHz bandwidth, or Band~2 with $100$\,GHz carrier frequency and $2$\,GHz bandwidth. The coherence block length is $\tau_c=200$, of which $\tau_p=8$ samples are used for uplink pilots. The maximum UE transmit power is $P^{\max}_{\rm UE}=0.2$\,W and the maximum fronthaul transmit power per AP is $P^{\max}_{\rm fh}=10$\,W. For the imperfect hardware case, we set $\kappa_{\rm ac}=\kappa_{\rm fh}=0.98$.

The APs and UEs are uniformly distributed in a square area of $1\,\text{km} \times 1\,\text{km}$, and the distances are computed by accounting for a vertical separation of $10$\,m between APs and UEs, and $20$\,m between APs and the CPU. The access-link channel gains follow a distance-dependent pathloss model given by $-32.4 - 20\log_{10}(f_c) - 31.9\log_{10}(d)$ (in dB), where the carrier frequency is $f_c=6$\,GHz, $d$ denotes the 3D distance in meters, and independent log-normal shadowing with standard deviation $8.2$\,dB is added \cite{3GPP5G}. The thermal noise power is computed based on a $50$\,MHz bandwidth and a noise figure of $5$\,dB. For the wireless fronthaul, the pathloss is modeled as $-32.4 - 20\log_{10}(f_c) - 21\log_{10}(d)$ (in dB), where $f_c$ is the carrier frequency in GHz (either $28$ or $100$ GHz), and shadowing with standard deviation $4$~dB is included. The fronthaul noise power is determined according to the corresponding bandwidth (either $400$\,MHz or $2$\,GHz) and a noise figure of $5$\,dB. The access channels are modeled as spatially correlated Rayleigh fading with local scattering, while the fronthaul channels follow a pure line-of-sight (LOS) model.

Fig.~1 shows the cumulative distribution function (CDF) of the uplink rate per UE when the wireless fronthaul operates in Band~1 (mmWave). The proposed fixed-point scheme provides a substantially improved fairness compared to the maximum-power baseline. In particular, the lower tail of the CDF is significantly shifted to higher rates, indicating a notable improvement for the worst-case UEs. This gain stems from the max-min nature of the proposed algorithm, which jointly balances the fronthaul and access-link SINRs while minimizing the fronthaul time expansion factor. In contrast, the maximum-power baseline leads to a wider spread in UE rates, with a non-negligible fraction of UEs experiencing very low data rates. Moreover, imperfect hardware mainly impacts the high-rate UEs, as seen from the right tail of the CDFs.

Fig.~2 presents the corresponding results for Band~2 (sub-THz). Due to the significantly larger fronthaul bandwidth, the overall rates increase for both schemes. However, in this regime, the performance gap between the proposed method and the maximum-power baseline becomes even more pronounced. This highlights that, despite the increased fronthaul capacity, the max-min optimization remains critical to efficiently utilize the available resources and to balance the system. In particular, without proper fronthaul-aware power control, strong UEs tend to dominate the system, whereas the proposed method enforces a more uniform rate distribution across UEs.

\begin{figure}[t!]
	\begin{center}
		\includegraphics[trim={0.6cm 0cm 1cm 0.6cm},clip,width=8cm]{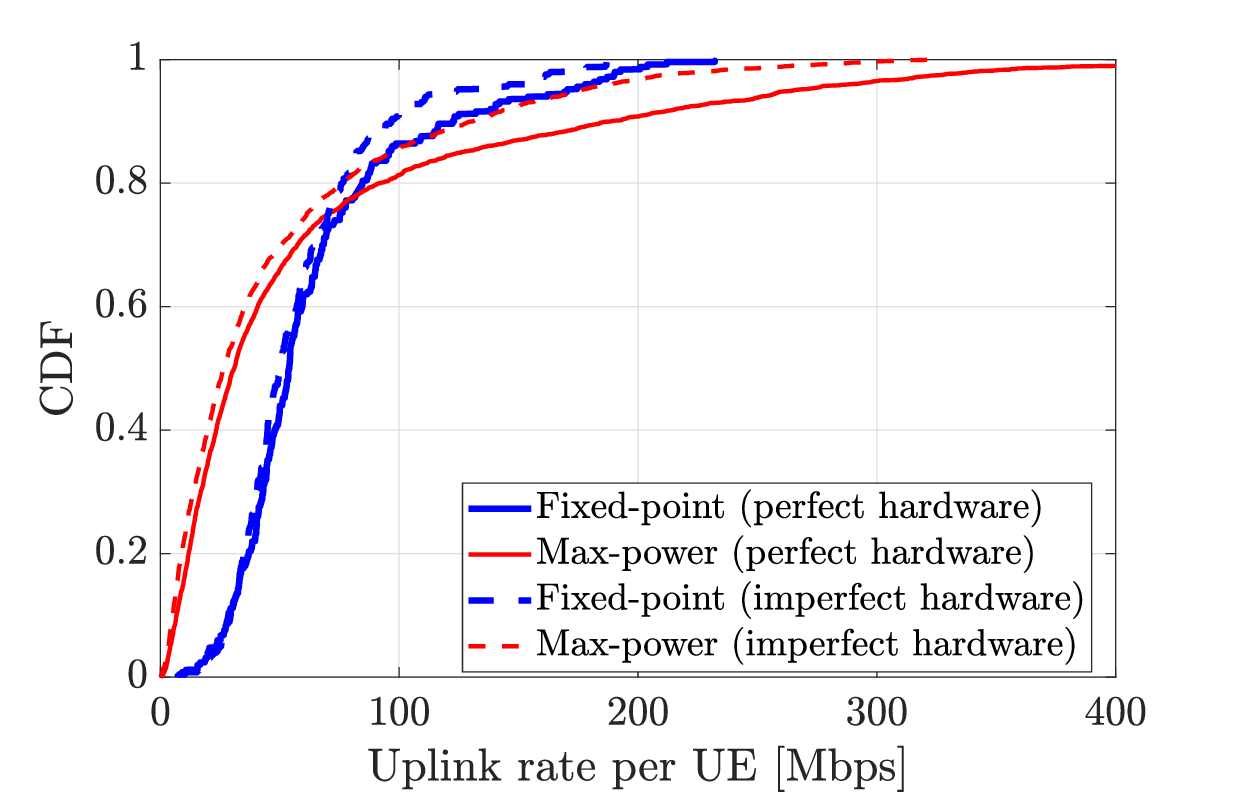}
             \vspace{-1mm}
		\caption{The CDF of uplink rate per UE for mmWave-based fronthaul.}
\label{fig1}
\vspace{-8mm}
	\end{center}
\end{figure}

\begin{figure}[t!]
	\begin{center}
		\includegraphics[trim={0.6cm 0cm 1cm 0.6cm},clip,width=8cm]{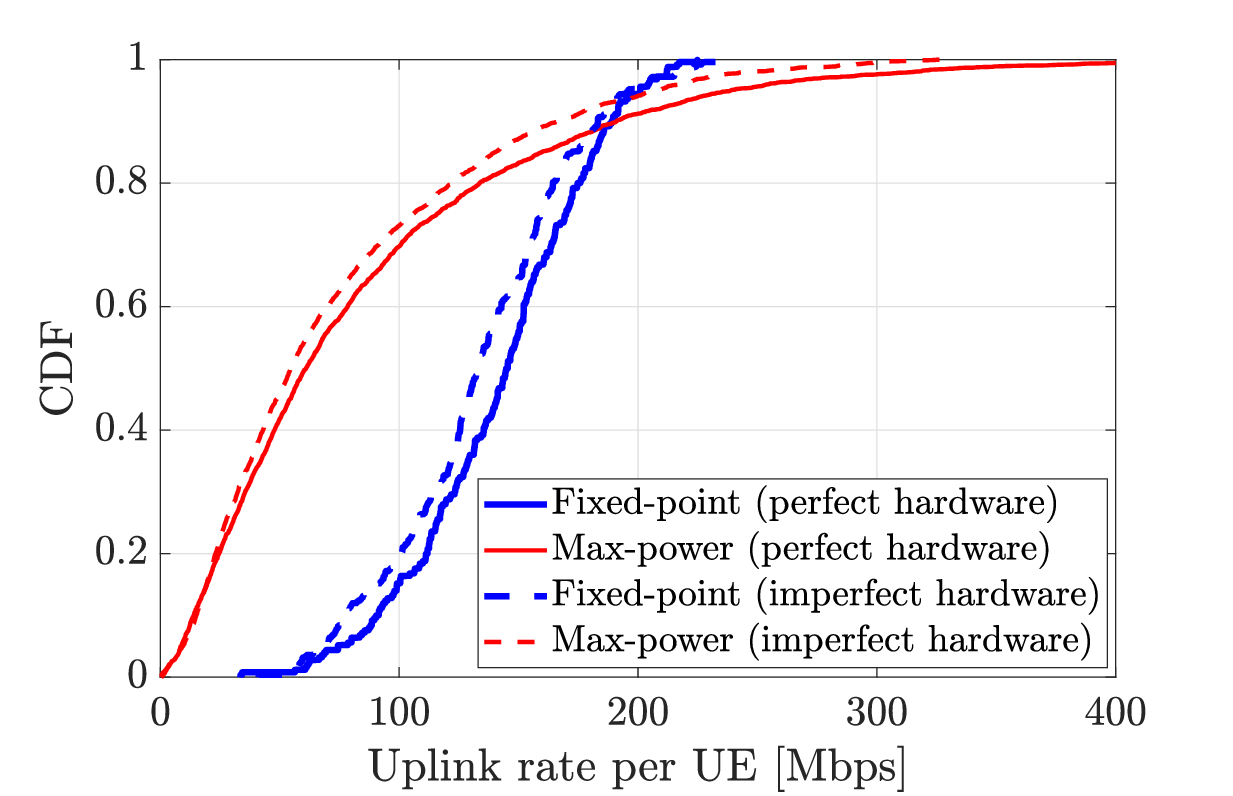}
        \vspace{-1mm}
		\caption{The CDF of uplink rate per UE for sub-THz-based fronthaul.}
\label{fig2}
\vspace{-8mm}
	\end{center}
\end{figure}
\section{Conclusion}

This paper investigated cell-free massive MIMO with wireless fronthaul under hardware impairments and a functional-split architecture. A unified analytical framework and a max-min fair resource allocation scheme were developed to jointly optimize access and fronthaul resources. The results show that accounting for fronthaul limitations is critical for achieving both high spectral efficiency and fairness, and that the performance gains of the proposed fixed-point method become even more pronounced at higher fronthaul bandwidths (e.g., sub-THz compared to mmWave).
\bibliographystyle{IEEEtran}
\bibliography{IEEEabrv,refs}

\end{document}